\documentclass{llncs}
\usepackage{setspace}
\usepackage{latexsym, amsmath}
\usepackage{slashbox}
\usepackage{amsfonts}
\usepackage{amsmath}
\usepackage{amssymb}
\usepackage{bbm}
\usepackage{verbatim}
\usepackage{fancybox}
\usepackage{tikz}
\usepackage[toc,page]{appendix}

\newcommand{\Z}{\mathbb{Z}}
\newcommand{\F}{\mathbb{F}}

\newcommand{\C}{\mathcal{C}}

\newenvironment{keywords}{
       \list{}{\advance\topsep by0.35cm\relax\small
       \leftmargin=1cm
       \labelwidth=0.35cm
       \listparindent=0.35cm
       \itemindent\listparindent
       \rightmargin\leftmargin}\item[\hskip\labelsep
                                     \bfseries Keywords:]}
     {\endlist}
\begin{document}
\date{today}
\title{On Determining Deep Holes of Generalized Reed-Solomon Codes
\thanks{The research is partially supported by NSF under grants CCF-0830522 and 
CCF-0830524 for Q.C. and J.Z., and 
by National Science Foundation of China
(11001170) and Ky and Yu-Fen Fan Fund Travel Grant from the AMS for J.L.}}
\author{Qi Cheng\inst{1} and Jiyou Li\inst{2} and Jincheng Zhuang \inst{1}}
\institute{
School of Computer Science\\
The University of Oklahoma\\
Norman, OK 73019, USA.\\
Email: {\tt qcheng@cs.ou.edu, jzhuang@ou.edu} \and
Department of Mathematics\\
Shanghai Jiao Tong University\\
Shanghai, PR China\\
Email: {\tt lijiyou@sjtu.edu.cn}
}

\maketitle \pagestyle{plain}
\begin{abstract}
For a linear code, deep holes are defined to be vectors
that are  further away from codewords than all other vectors.
The problem of deciding whether a received word is a deep hole
for generalized Reed-Solomon codes is proved to be co-NP-complete \cite{GV05}\cite{CM07}.
For the extended Reed-Solomon codes $RS_q(\F_q,k)$, a conjecture was made to classify deep holes in \cite{CM07}.
Since then a lot of effort has been made to prove the conjecture, or its various forms.
In this paper, we classify deep holes completely for generalized Reed-Solomon codes $RS_p (D,k)$, where
$p$ is a prime, $|D| > k \geqslant \frac{p-1}{2}$. Our techniques  are built on the idea of
deep hole trees, and several results concerning  the Erd{\"o}s-Heilbronn conjecture.
\end{abstract}
\begin{keywords}
Reed-Solomon code, deep hole, deep hole tree, Erd{\"o}s-Heilbronn conjecture.
\end{keywords}
\section{Introduction}
Reed-Solomon codes are of special interest and importance both in theory and practice of error-correcting.
\begin{definition}
Let $\F_q$ be a finite field with q elements and characteristic $p$.
Let $D=\{\alpha_1,\ldots,\alpha_n\}\subseteq \F_q$ be the evaluation set and $v_i\in\F_q^*, 1\leqslant i \leqslant n,$ be the column multipliers. The
set of codewords of the generalized Reed-Solomon code $RS_q(D,k)$ of
length $n$ and dimension $k$ over $\F_q$ is defined as
$$RS_q(D,k)=\{(v_1f(\alpha_1),\ldots,v_nf(\alpha_n))\in \F_q^n \mid f(x)\in \F_q[x],
deg(f)\leqslant k-1\}.$$
\end{definition}
We will write generalized Reed-Solomon codes as GRS codes for short in the sequel.
If $D=\F_q^*$, it is called \textit{primitive}. If $D=\F_q$,
it is called a \textit{singly-extended} GRS code. A GRS code is called \textit{normalized} if its column multipliers are all equal to 1.
In this paper, we will work on the normalized GRS without loss of generality. 

The encoding algorithm of the GRS code can be
described by the linear map $\varphi: \F_q^k\rightarrow \F_q^n$, in which a message $(a_1,\ldots,a_k)$ is mapped to a codeword $(f(\alpha_1),\ldots,f(\alpha_n))$,
where $f(x)=a_kx^{k-1}+a_{k-1}x^{k-2}+\cdots+a_1\in \F_q[x].$

The \textit{Hamming distance} between two words is the number of their distinct coordinates. The \textit{error distance} of a received word $u\in\F_q^n$ to the code is defined as
its minimum Hamming distance to codewords. The \textit{minimum distance} of a code, which is denoted by $d$, is the smallest distance between any two distinct codewords of the code. The \textit{covering radius} of a code is the maximum distance from any vector in $\F_q^n$ to the nearest codeword. A \textit{deep hole} is a vector achieving the covering radius.
A linear code $[n,k]_q$ is called \textit{maximum distance separable} (in shot, MDS) if it attains the Singleton bound, i.e., $d=n-k+1$. GRS code is
a linear MDS code, and its minimum distance is known to be $n-k+1$ and the covering radius is known to
be $n-k$. Thus for the GRS code, $u$ is a deep hole if $d(u, RS_q(D,k))=n-k.$
A linear code can be represented by a generator matrix. In this paper,
we assume that the rows of a generator matrix form a basis for the code.

\subsection{Related work}
Efforts have been made to obtain an efficient decoding algorithm for
GRS codes. Given a received word $u\in \F_q^n$, if the
error distance is smaller than $n-\sqrt{nk}$, then the list decoding
algorithm of Sudan \cite{Sudan97} and Guruswami-Sudan \cite{GS99}
solves the decoding in polynomial time. However, in general, the
maximum likelihood decoding of GRS codes is NP-hard
\cite{GV05}.

We would like to determine all the deep holes of the code. To this end, given a
received word $u=(u_1,u_2,\ldots,u_n)\in \F_q^n$, we consider the
following Lagrange interpolating polynomial
$$u(x)=\sum_{i=1}^nu_i\frac{\prod_{j\neq i}(x-\alpha_j)}{\prod_{j\neq i}(\alpha_i-\alpha_j)}\in \F_q[x],$$
where $D=\{\alpha_1,\ldots,\alpha_n\}$ is the evaluation set. The Lagrange
interpolating polynomial is the only polynomial in $\F_q[x]$ of
degree less than $n$ that satisfies $u(\alpha_i)=u_i,1\leqslant i
\leqslant n.$ In this paper, we say that a function \textit{$u(x)$ generates} a vector $u\in\F_q^n$ if
$u=(u(\alpha_1),u(\alpha_2),\ldots,u(\alpha_n))$. We have the
following conclusions:
\begin{itemize}
\item If $deg(u)\leqslant k-1$, then $u\in RS_q(D,k)$ by definition and
$d(u,RS_q(D,k))=0$.
\item If $deg(u)=k$, then it can be shown that $u$ is a deep hole by the following proposition \cite{lw08},
i.e., $d(u, RS_q(D,k))=n-k$.
\end{itemize}
\begin{proposition} (\cite{lw08})
For $k \leqslant deg(u) \leqslant n-1$, we have the inequality
$$n-deg(u)\leqslant d(u, RS_q(D,k)) \leqslant n-k.$$
\end{proposition}

When the degree of $u(x)$ becomes larger than $k$, the situation
becomes complicated for GRS codes. However, in
the case of (singly-)extended GRS codes, the situation seems to
be much simpler. Cheng and Murray \cite{CM07} conjectured in 2007
that the vectors generated by polynomial of degree $k$ are the only
possible deep holes.
\begin{conjecture}\label{conj1} (\cite{CM07})
A word $u$ is a deep hole of $RS_q(F_q, k)$ if and only if $\deg(u)=k.$
\end{conjecture}

There is an analogous conjecture for deep holes of primitive Reed-Solomon codes by Wu and Hong \cite{WH12}.
\begin{conjecture}\label{conj2} (\cite{WH12})
A word $u$ is a deep hole of $RS_q(F_q^*, k)$ if and only if:
$$u(x)=ax^k+f_{\leqslant k-1}(x), a\neq 0;$$
or
$$u(x)=bx^{q-2}+f_{\leqslant k-1}(x), b\neq 0;$$
where $f_{\leqslant k-1}(x)$ denotes a polynomial with degree not larger
than $k-1$.
\end{conjecture}

Cheng and Murray \cite{CM07} got the first result by reducing the problem to the existence of rational points on a hypersurface over $\F_q$.
\begin{theorem} \cite{CM07}
Let $u\in \F_q^{q}$ such that $1\leqslant d:=\deg(u)-k\leqslant q-1-k$. If $q\geqslant \max(k^{7+\epsilon},d^{\frac{13}{3}+\epsilon})$
for some constant $\epsilon>0$, then $u$ is not a deep hole.
\end{theorem}

Following a similar approach of  Cheng-Wan \cite{CW07}, Li and Wan \cite{LJW08} improved the above result with Weil's character sum estimate.
\begin{theorem} \cite{LJW08}
Let $u\in \F_q^q$ such that $1\leqslant d:=\deg(u)-k\leqslant q-1-k$. If
\[
q>\max((k+1)^2,d^{2+\epsilon}),k>(\frac{2}{\epsilon}+1)d+\frac{8}{\epsilon}+2
\]
for some constant $\epsilon>0$, then $u$ is not a deep hole.
\end{theorem}

Then Liao \cite{Liao11} proved the following result:
\begin{theorem} \cite{Liao11}
Let $r\geqslant 1$ be an integer. For any received word $u\in \F_q^q, r\leqslant d:=\deg(u)-k\leqslant q-1-k$, if
$$
q>\max(2\binom{k+r}{2}+d,d^{2+\epsilon}),k>(\frac{2}{\epsilon}+1)d+\frac{2r+4}{\epsilon}+2
$$
for some constant $\epsilon>0$, then $d(u,RS_q(\F_q,k))\leqslant q-k-r$, which implies that $u$ is not a deep hole.
\end{theorem}

Antonio Cafure ect. \cite{CMP12} proved the following result with tools of algebraic geometry:
\begin{theorem}\cite{CMP12}
Let $u\in \F_q^q$ such that $1\leqslant d:=\deg(u)-k\leqslant q-1-k$. If
\[
q>\max((k+1)^2,14d^{2+\epsilon}),k>(\frac{2}{\epsilon}+1)d,
\]
for some constant $\epsilon>0$, then $u$ is not a deep hole.
\end{theorem}

Using Weil's character sum estimate and Li-Wan's new sieve \cite{LW10} for distinct coordinates counting, Zhu and Wan \cite{ZW12}
showed the following result:
\begin{theorem} \cite{ZW12}
Let $r\geqslant 1$ be an integer. For any received word $u\in \F_q^q, r\leqslant d:=\deg(u)-k\leqslant q-1-k$, there are positive
constants $c_1$ and $c_2$ such that if
$$
d<c_1q^{1/2},(\frac{d+r}{2}+1)\log_2(q)<k<c_2q,
$$
then $d(u,RS_q(\F_q,k))\leqslant q-k-r$.
\end{theorem}

The deep hole problem for Reed-Solomon codes are also closely related to
the famous MDS conjecture in coding theory. On one hand, GRS codes are MDS codes.
On the other hand, it is known that all long enough MDS codes are essentially GRS codes.
Following the notation of \cite{RL89}, let $N_{min}(k,q)$ be the minimal integer, if any, such that every $[n,k]$ MDS code over $GF(q)$ with
$n>N_{min}(k,q)$ is GRS and be $q+2$ if no such integer exists. For the case of $k=3$, Segre \cite{Segre55} obtained the following result:
\begin{theorem} \cite{Segre55} \label{B.Segre}
If $q$ is odd, every $[n,3]$ MDS code over $GF(q)$ with $q-\dfrac{\sqrt{q}-7}{4} < n \leqslant q+1$ is GRS.
\end{theorem}

When $q=p$ is a prime, Voloch \cite{Voloch90} obtained the following result:
\begin{theorem}\cite{Voloch90}\label{Voloch}
If $p$ is an odd prime number, every $[n,3]$ MDS code over $GF(p)$ with $p-\dfrac{p}{45}+2< n \leqslant p+1$ is GRS.
\end{theorem}

Further, there is a relation for $N_{min}(k+1,q)$ and $N_{min}(k,q)$ \cite{RL89} as follows:
\begin{lemma}\cite{RL89}\label{RL}
For $3\leqslant k \leqslant q-2$, we have
$$N_{min}(k+1,q)\leqslant N_{min}(k,q)+1.$$
\end{lemma}

Simeon Ball \cite{Ball12} showed the following result:
\begin{theorem}\cite{Ball12}\label{Simeon}
Let $S$ be a set of vectors of the vector space $\F_q^k$, with the
property that every subset of $S$ of size $k$ is a basis. If
$|S|=q+1$ and $k\leqslant p$ or $3\leqslant q-p+1\leqslant
k\leqslant q-2$, where
 $p$ is the characteristic of $\F_q$, then $S$ is equivalent to the following set:
$$\{(1,\alpha,\alpha^2,\ldots,\alpha^{k-1})\mid \alpha\in \F_q\}\cup\{(0,\ldots,0,1)\}.$$
\end{theorem}

\subsection{Our result}

In this paper, we classify the deep holes in many cases.
Firstly, we show:
\begin{theorem}\label{mainthm}
Let $p>2$ be a prime number, $k\geqslant\frac{p-1}{2}, D=\{\alpha_1,\alpha_2,\ldots,\alpha_n\}$ with $k < n\leqslant p$.
The only deep holes of $RS_p(D,k)$ are generated by functions which are equivalent to the following:
\[
f(x)=x^k,\quad f_{\delta}(x)=\frac{1}{x-\delta},
\]
where $\delta\in \F_p\setminus D$. Here two  functions $f(x)$ and
$g(x)$ are equivalent if and only if there exists $a\in \F_p^*$ and
$h(x)$ with degree less than $k$ such that
\[
g(x)=af(x)+h(x).
\]
\end{theorem}
Our techniques  are built on the idea of deep hole trees, and
several results concerning  the Erd\"os-Heilbronn conjecture. 
We also show the following theorem based on some results of finite geometry.
\begin{theorem}\label{geometrythm}
Given a finite filed $\F_q$ with characteristic $p>2$, we have
\begin{itemize}
\item If $k+1\leqslant p$ or $3\leqslant q-p+1\leqslant k+1\leqslant q-2$, then Conjecture \ref{conj1} is true.
\item If $3\leqslant k < \dfrac{\sqrt{q}+1}{4}$, then Conjecture \ref{conj2} is true.
\item If $3\leqslant k < \dfrac{p}{45}$, where $q=p$ is prime, then Conjecture \ref{conj2} is true.
\end{itemize}
\end{theorem}

This paper is organized as follows: Section 2 presents some preliminaries; 
Section 3 describes the idea of the deep hole tree; Section 4 demonstrates 
the proof of Theorem \ref{mainthm}; Section 5 gives the proof of Theorem 
\ref{geometrythm}.

\section{Preliminaries}
\subsection{A criterion for deep holes of linear MDS codes}
By definition, deep holes of a linear code are words that has a maximum distance to the code.
In the case of linear MDS codes, there is another way to characterize the deep hole as
follows, which connects the concept of deep holes with the MDS codes.
The following is well known:
\begin{proposition}\label{prop}
Let $\F_q$ be a finite field with characteristic $p$. Suppose $G$ is
a generator matrix for a linear MDS code $C=[n,k]_q$ with covering radius $\rho=n-k$, then $u\in \F_q^n$ is a deep hole of $C$ if and only if
$$
G'=\left[
\begin{array}{c}
 G\\ \hline
 u
\end{array}\right]
$$
generates another linear MDS code.
\end{proposition}
A proof is included in Appendix A for the sake of completeness.
\subsection{Some additive combinatorics results}
In this section, we introduce some additive combinatorics results that we will use later.
The first theorem is about the estimation of the size of restricted sum sets, which is first proved by Dias da Silva and Hamidoune \cite{SH94}. Then Alon et al. \cite{ANR96}
gave a simple proof using the polynomial method.
\begin{theorem}\cite{SH94,ANR96}\label{SH}
Let $\F$ be a field with characteristic $p$ and $n$ be a positive integer. Then for any finite subset $S\subset \F$ we have
$$ |n ^\wedge  S |\geqslant \min\{p, n|S|-n^2+1\},$$
where $n ^\wedge S$ denotes the set of all sums of $n$ distinct elements of $S$.
\end{theorem}

Brakemeier \cite{B78} and Gallardo et al. \cite{GGP99} established the following theorem:
\begin{theorem}\cite{B78,GGP99}\label{2sum}
Let $n$ be a positive integer and $S\subset \Z/n\Z$. If $|S|>\frac{n}{2}+1$, then
\[
2^\wedge S = \Z/n\Z,
\]
where $2^ \wedge S$ denotes the set of all sums of $2$ distinct elements of $S$.
\end{theorem}

Hence we have the following corollary:
\begin{corollary}\label{2product}
Let $\F_p$ be a prime finite field, $S\subset \F_p^*$. If $|S|>\frac{p+1}{2}$, then each element of $\F_p^*$ is the product of two distinct
elements of $S$.
\end{corollary}
\begin{proof}
Let $g$ be a generator of $\F_p^*$. Let
\[
S'=\{e|g^e\in S\}\subset \Z/(p-1)\Z.
\]
For any given element $\alpha=g^a\in\F_p^*$, we need to show that there exist two distinct
elements $b\neq c$ such that
\[
g^a=g^bg^c,
\]
where $b,c\in S'$.
This is equivalent to
\[
a = b + c,
\]
which follows from Theorem \ref{2sum}.
\end{proof}

\section{Construction of the deep hole tree}
Let $\F_q = \{ \alpha_1, \alpha_2, \cdots, \alpha_{q}=0\} $.
The polynomials in $\F_q [x]$ of degree less than $q$ form a $\F_q$-linear
space,
with a basis
\[
\{1,x,\ldots,x^{k-1},\prod_{i=1}^{k}(x-\alpha_i),\ldots,\prod_{i=1}^{q-1}(x-\alpha_i)\}. 
\]
Given a polynomial $f(x)\in \F_q[x]$ with degree $q-1$
we have
\[
f(x)=l(x)+c_1\prod_{i=1}^{k}(x-\alpha_i)+\cdots+c_{q-k}\prod_{i=1}^{q-1}(x-\alpha_i),
\]
where $l(x)$ is of degree less than $k$, we want to determine when $f(x)$ generates a deep hole. By Proposition \ref{prop}, $f(x)$ generates a deep hole of $RS_q(\F_q,k)$ if and only if
$$
G'=\left[
\begin{array}{c}
 G\\ \hline
 u
\end{array}\right]
$$
generates an MDS code, where $G$ is the generator matrix of $RS_q(\F_q,k)$, and $u=(f(\alpha_1),\ldots,f(\alpha_q))$.

Observe that the function, which generates a deep hole for $RS_q(D_2,k)$, also generates a deep hole for $RS_q(D_1,k)$ if $D_1\subset D_2$.
Instead of considering the deep holes for $RS_q(\F_q,k)$ at the first step, we propose to consider a smaller evaluation set at the beginning and make it increase gradually.
The advantage of doing so is that we can prune as we increase the evaluation set instead of exhausting search $(c_1,\ldots,c_{q-k})$. To be more precise,
firstly we determine $c_1$ over $D_1=\{\alpha_1,\ldots,\alpha_{k+1}\}$, then we determine $c_2$ over $D_2=\{\alpha_1,\ldots,\alpha_{k+2}\}$  based on the knowledge of $c_1$, so on and so forth. We present the result as a tree, which we will call a \textit{deep hole tree} in the sequel. 

\begin{remark} \label{expdeep}
Wu and Hong \cite{WHarx12} showed that if $D=\F_q\setminus\{\beta_1,\ldots,\beta_l\}$ then $f_{\beta_i}(x)=\frac{1}{x-\beta_i}$ generates a deep hole for $RS_q(D,k)$, where $1\leqslant i \leqslant l$. We can also deduce this from Proposition \ref{prop}. For convenience, we will call these deep holes and deep holes generated by a function of degree $k$ \textit{expected deep holes}.
\end{remark}

Motivated by Remark \ref{expdeep}, firstly we construct the \textit{expected deep hole tree} as follows:
\begin{itemize}
\item The root node is $1$ without loss of generality, i.e., $c_1=1$.
\item There are $p-k-1$ branches of the tree, each with distinct length in $[2,p-k]$. And we designate the sequence of nodes in a branch with length $l$ as $b_l$.
    \begin{itemize}
    \item If $l=p-k$, then $b_{p-k}=\{0,\ldots,0\}$.
    \item If $2\leqslant l \leqslant p-k-1$, then $b_l=(c_1,\ldots,c_l)$, where $f=\frac{1}{x-\alpha_{l+1}}$ is equivalent to $c_1\prod_{i=1}^{k}(x-\alpha_i)+\cdots+c_{l}\prod_{i=1}^{k+l-1}(x-\alpha_{i})$.
    \end{itemize}
\end{itemize}

\begin{proposition}\label{exptree}
The expected deep hole tree is a part of the full deep hole tree.
\end{proposition}
\begin{proof}
This follows from Remark \ref{expdeep}.
\end{proof}

Now we can construct the full deep hole tree based on the expected deep hole tree.
\begin{itemize}
\item The root node is $1$ without loss of generality, i.e., $c_1=1$.
\item The children $\{c_{i+1}\}$ of a node $c_i,1\leqslant i \leqslant q-k-1$ are defined as follows: given the ancestors $(c_1,\ldots,c_i)$, for $\gamma\in GF(q)$, if $\gamma$ is the
child of $c_i$ in the expected deep hole tree, then keep it; otherwise, if
\[
c_1\prod_{i=1}^{k}(x-\alpha_i)+\cdots+c_{i}\prod_{i=1}^{k+i-1}(x-\alpha_i)+\gamma\prod_{i=1}^{k+i}(x-\alpha_i)
\]
satisfies the property of the function which generates a deep hole as in Proposition \ref{prop}, then $\gamma$ is a child of $c_i$.
\end{itemize}
That is, we keep the nodes of the expected deep hole tree and add additional ones if necessary. Now we illustrate the procedure to construct the deep hole tree by one example.

\textbf{Example 1.} Let $p=7,k=2$. The evaluation set is ordered such that $\alpha_i=i,1\leqslant i \leqslant 7$.

(1) The expected deep hole tree is as follows:

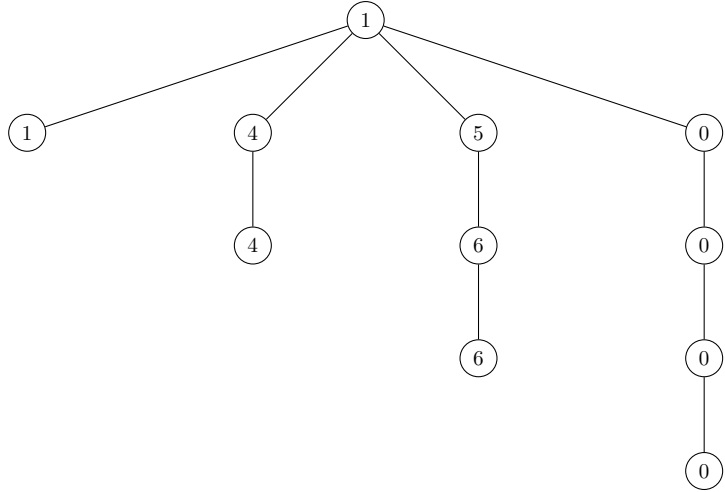
\begin{tikzpicture}[level/.style={sibling distance=30mm/#1},scale=1.0, every node/.style={scale=0.8}]
\node [circle,draw] (z){$1$}
  child {node [circle,draw] (j) {$1$}
     }
  child {node [circle,draw] (j) {$4$}
      child  {node [circle,draw] {$4$}}
     }
  child {node [circle,draw] (j) {$5$}
       child  {node [circle,draw] {$6$}
            child  {node [circle,draw] {$6$}}
        }
     }
  child {node [circle,draw] (a) {$0$}
      child {node [circle,draw] (b) {$0$}
        child {node [circle,draw] {$0$}
          child {node [circle,draw] (d) {$0$}
             }
         }
        }
       }
     ;
\end{tikzpicture}

The root is corresponding to the evaluation set $D_1=\{1,2,3\}$. The expected deep holes are generated by functions equivalent to
$f=\prod_{i=1}^{2}(x-i)$. In depth 2, the evaluation set is $D_2=\{1,2,3,4\}$. One of the expected deep holes is generated by the function $f=\prod_{i=1}^{2}(x-i)+\prod_{i=1}^{3}(x-i)$,
which is equivalent to $f=\frac{1}{x-5}$. In depth 3, the evaluation set is $D_3=\{1,2,3,4,5\}$. One of the expected deep holes is generated by the function $f=\prod_{i=1}^{2}(x-i)+4\prod_{i=1}^{3}(x-i)+4\prod_{i=1}^{4}(x-i)$, which is equivalent to $f=\frac{1}{x-6}$. In depth 4, the evaluation set is $D_4=\{1,2,3,4,5,6\}$.
One of the expected deep holes is generated by the function $f=\prod_{i=1}^{2}(x-i)+5\prod_{i=1}^{3}(x-i)+6\prod_{i=1}^{4}(x-i)+6\prod_{i=1}^{5}(x-i)$, which is equivalent to $f=\frac{1}{x}$. In depth 5, the evaluation set is $D_5=\{1,2,3,4,5,6,7\}$. One of the expected deep holes is generated by the function $f=\prod_{i=1}^{2}(x-i)$.

(2) The full deep hole tree is as follows:

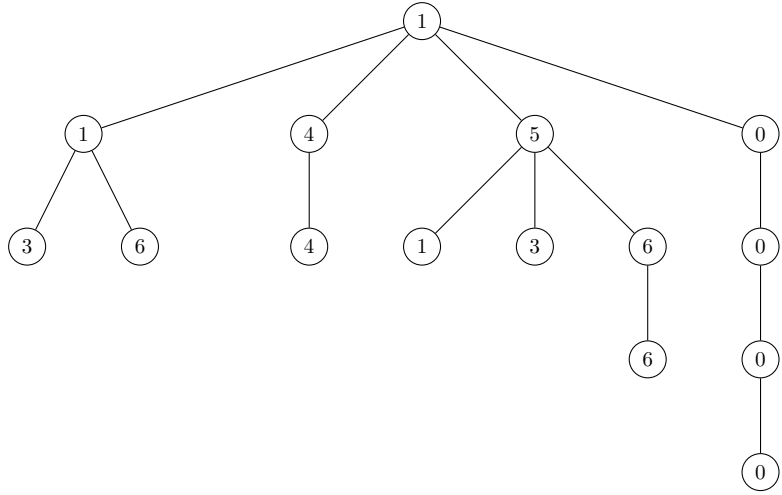
\begin{tikzpicture}[level/.style={sibling distance=30mm/#1}, scale=1.0, every node/.style={scale=0.8}]
\node [circle,draw] (z){$1$}
  child {node [circle,draw] (j) {$1$}
      child  {node [circle,draw] {$3$}}
      child  {node [circle,draw] {$6$}}
    }
  child {node [circle,draw] (j) {$4$}
      child  {node [circle,draw] {$4$}}
     }
  child {node [circle,draw] (j) {$5$}
      child  {node [circle,draw] {$1$}}
      child  {node [circle,draw] {$3$}}
      child  {node [circle,draw] {$6$}
            child  {node [circle,draw] {$6$}}
        }
     }
   child {node [circle,draw] (a) {$0$}
       child {node [circle,draw] (b) {$0$}
         child {node [circle,draw] {$0$}
           child {node [circle,draw] (d) {$0$}
             }
          }
         }
        }
     ;
\end{tikzpicture}

Note that there are more nodes here than the expected ones. For example, in depth 3, there is an additional deep hole generated by the function $f=\prod_{i=1}^{2}(x-i)+\prod_{i=1}^{3}(x-i)+3\prod_{i=1}^{4}(x-i)$. Also, there is an additional deep hole generated by the function
$f=\prod_{i=1}^{2}(x-i)+5\prod_{i=1}^{3}(x-i)+\prod_{i=1}^{4}(x-i)$.

\section{Proof of Theorem \ref{mainthm}}
The basic idea of the proof of Theorem \ref{mainthm} is reducing the problem to some additive number theory problems. We first present several lemmas.
\begin{lemma} \label{basis}
In depth $d=2$, the nodes are the same in both the expected deep hole tree and full deep hole tree.
\end{lemma}
\begin{proof}
We need to show that in depth $d=2$, the nodes are the same in both the expected deep hole tree and full deep hole tree.

In depth $d=2$, the evaluation set is $D=\{\alpha_1,\alpha_2,\ldots,\alpha_{k+2}\}$, the candidate generating function is
$$f(x)=\prod_{i=1}^{k}(x-\alpha_i)+c_2\prod_{i=1}^{k+1}(x-\alpha_i),\quad c_2\in \F_p.$$

Designate the set of nodes in depth $2$ of the expected deep hole tree as $S$.
Firstly, we show that if $c_2\in S$, then $f(x)$ generates a deep hole.
This follows from Theorem \ref{exptree}. Note that there are $p-(k+1)$ of them. Now we show that they are all distinct.
Considering the following square matrix
\[
G=
\begin{bmatrix}
1           & 1          & \cdots      & 1         \\
\vdots      & \vdots     & \ddots      & \vdots  \\
\alpha_1^{k-1} & \alpha_2^{k-1} & \cdots & \alpha_{k+2}^{k-1} \\
\frac{1}{\alpha_1-\delta_1}   & \frac{1}{\alpha_2-\delta_1}   & \cdots      & \frac{1}{\alpha_{k+2}-\delta_1} \\
\frac{1}{\alpha_1-\delta_2}   & \frac{1}{\alpha_2-\delta_2}   & \cdots      & \frac{1}{\alpha_{k+2}-\delta_2}
\end{bmatrix},
 \]
we compute
\[
\prod_{i=1}^{k+2}(\alpha_i-\delta_1)\det(G) =
\begin{vmatrix}
\alpha_1-\delta_1                             & \cdots        & \alpha_{k+2}-\delta_1                               \\
\alpha_1*(\alpha_1-\delta_1)                  & \cdots        & \alpha_{k+2}*(\alpha_{k+2}-\delta_1)                   \\
\vdots                                        & \ddots        & \vdots                                           \\
\alpha_1^{k-1}*(\alpha_1-\delta_1)            & \cdots        & \alpha_{k+2}^{k-1}*(\alpha_{k+2}-\delta_1)                \\
1                                             & \cdots        & 1                                               \\
\frac{\alpha_1-\delta_1}{\alpha_1-\delta_2}   & \cdots        & \frac{\alpha_{k+2}-\delta_1}{\alpha_{k+2}-\delta_2}
\end{vmatrix}
\]
\[
=
\begin{vmatrix}
\alpha_1                                      & \cdots        & \alpha_{k+2}                              \\
\alpha_1^2                                    & \cdots        & \alpha_{k+2}^2                  \\
\vdots                                        & \ddots        & \vdots                                           \\
\alpha_1^k                             & \cdots        & \alpha_{k+2}^k                \\
1                                             & \cdots        & 1                                               \\
\frac{\alpha_1-\delta_1}{\alpha_1-\delta_2}   & \cdots        & \frac{\alpha_{k+2}-\delta_1}{\alpha_{k+2}-\delta_2}
\end{vmatrix}
\]
\[
\qquad \quad =
(-1)^k
\begin{vmatrix}
1                                             & \cdots        & 1                                               \\
\alpha_1                                      & \cdots        & \alpha_{k+2}                              \\
\alpha_1^2                                    & \cdots        & \alpha_{k+2}^2                  \\
\vdots                                        & \ddots        & \vdots                                           \\
\alpha_1^k                             & \cdots        & \alpha_{k+2}^k                \\
\frac{\alpha_1-\delta_1}{\alpha_1-\delta_2}   & \cdots        & \frac{\alpha_{k+2}-\delta_1}{\alpha_{k+2}-\delta_2}
\end{vmatrix}
\]
\[
\qquad \qquad \qquad  =
(-1)^k
\begin{vmatrix}
1                                             & \cdots        & 1                                               \\
\alpha_1                                      & \cdots        & \alpha_{k+2}                              \\
\alpha_1^2                                    & \cdots        & \alpha_{k+2}^2                  \\
\vdots                                        & \ddots        & \vdots                                           \\
\alpha_1^k                                    & \cdots        & \alpha_{k+2}^k                \\
1+\frac{\delta_2-\delta_1}{\alpha_1-\delta_2} & \cdots        & 1+\frac{\delta_{2}-\delta_1}{\alpha_{k+2}-\delta_2}
\end{vmatrix}
\]
\[
\qquad \qquad \qquad \qquad  =
(-1)^k(\delta_2-\delta_1)
\begin{vmatrix}
1                                             & \cdots        & 1                                               \\
\alpha_1                                      & \cdots        & \alpha_{k+2}                              \\
\alpha_1^2                                    & \cdots        & \alpha_{k+2}^2                  \\
\vdots                                        & \ddots        & \vdots                                           \\
\alpha_1^k                                    & \cdots        & \alpha_{k+2}^k                \\
\frac{1}{\alpha_1-\delta_2}                   & \cdots        & \frac{1}{\alpha_{k+2}-\delta_2}
\end{vmatrix}
\neq 0.
\]
The last step follows from the same argument as above. Hence we prove that $c_2$ are different for $f_1=\frac{1}{x-\delta_1}$ and $f_2=\frac{1}{x-\delta_2}$.

Next, we show that if $c_2\notin S$ then $f(x)$ does not generate a deep hole. Consider the following matrix
$$
G=
\begin{bmatrix}
1                            & 1                 & \cdots                       & 1  \\
\alpha_1                     & \alpha_2          & \cdots                       & \alpha_{k+2}\\
\vdots                       & \vdots            & \ddots                       & \vdots   \\
\alpha_1^{k-1}               & \alpha_2^{k-1}    & \cdots                       & \alpha_{k+2}^{k-1} \\
f(\alpha_1)                  & f(\alpha_2)       & \cdots                       & f(\alpha_{k+2})
\end{bmatrix},
$$
where $f(i)=0,1\leqslant i\leqslant k, f(\alpha_{k+1})=\prod_{i=1}^{k}(\alpha_{k+1}-\alpha_i), f(\alpha_{k+2})=\prod_{i=1}^{k}(\alpha_{k+2}-\alpha_i)+c_{2}\prod_{i=1}^{k+1}(\alpha_{k+2}-\alpha_i) $.

\noindent
\textbf{Case 1.} If
\[
\begin{aligned}
f(\alpha_{k+2}) &= \prod_{i=1}^{k}(\alpha_{k+2}-\alpha_i)+c_2\prod_{i=1}^{k+1}(\alpha_{k+2}-\alpha_i) \\
                &= \prod_{i=1}^{k}(\alpha_{k+2}-\alpha_i)[1+c_2(\alpha_{k+2}-\alpha_{k+1})] \\
                &= 0,
\end{aligned}
\]
i.e., $c_2=\frac{1}{\alpha_{k+1}-\alpha_{k+2}}$, then there are $k+1$ columns of $G$ which are linearly dependent.
Thus $f(x)$ does not generate a deep hole in this case.

\noindent
\textbf{Case 2.} Suppose $f(\alpha_{k+2})\neq 0$. For any $k-1$ elements $\{\beta_1,\ldots,\beta_{k-1}\}\subset \{\alpha_1,\ldots,\alpha_k\}$, consider the submatrix
$$
G'=
\begin{bmatrix}
1                            & \cdots            & 1                            & 1                    & 1  \\
\beta_1                      & \cdots            & \beta_{k-1}                  & \alpha_{k+1}          & \alpha_{k+2}  \\
\vdots                       & \ddots            & \vdots                       & \vdots               & \vdots         \\
\beta_1^{k-1}                & \cdots            & \beta_{k-1}^{k-1}            & \alpha_{k+1}^{k-1}   & \alpha_{k+2}^{k-1} \\
0                            & \cdots            & 0                            & f(\alpha_{k+1})      & f(\alpha_{k+2})
\end{bmatrix}.
$$
Thus $\det(G')=0$ is equivalent to
\[
f(\alpha_{k+1})\prod_{i=1}^{k-1}(\alpha_{k+2}-\beta_i)=f(\alpha_{k+2})\prod_{i=1}^{k-1}(\alpha_{k+1}-\beta_i),
\]
that is,
\[
\begin{aligned}
\frac{f(\alpha_{k+2})}{f(\alpha_{k+1})} &= \prod_{i=1}^{k-1}\frac{\alpha_{k+2}-\beta_i}{\alpha_{k+1}-\beta_i} \\
&= \prod_{i=1}^{k-1}(1+\frac{\alpha_{k+2}-\alpha_{k+1}}{\alpha_{k+1}-\beta_i}).
\end{aligned}
\]
Hence for each subset of $\{\beta_1,\ldots,\beta_{k-1}\}\subset \{\alpha_1,\ldots,\alpha_k\}$, there is a unique $c_2$ such that $\det(G')=0$.

In total, there are $k+1$ elements of candidate $c_2$ such that the corresponding $f(x)$ does not generate a deep hole. This implies that if $c_2\notin S$ then $f(x)$ does not generate a deep hole.

In conclusion, in depth $d=2$, the nodes in the full deep hole tree are exactly those in the expected deep hole tree.
\end{proof}

\begin{lemma} \label{case_1}
Let $p$ be an odd prime, $k\geqslant \frac{p-1}{2}, d\geqslant 2$ be a positive integer and $D_d=\{\alpha_1,\ldots,\alpha_{k+d}\}\subset\F_p,\delta\in \F_p\setminus D_d$.
For any $\gamma\in \F_p$, there exists a subset $\{\beta_1,\ldots,\beta_k\}\subset D_d$ such that the matrix
\[
A=
\begin{bmatrix}
1                           & \cdots            & 1                            & 1  \\
\beta_1                     & \cdots            & \beta_k                      & \delta\\
\vdots                      & \ddots            & \vdots                       & \vdots   \\
\beta_1^{k-1}               & \cdots            &  \beta_{k}^{k-1}             & \delta^{k-1} \\
\frac{1}{\beta_1-\delta}    & \cdots            & \frac{1}{\beta_k-\delta}     & \gamma
\end{bmatrix},
\]
is singular.
\end{lemma}
\begin{proof}
Note that $\det(A)=\det(A')+\det(A'')$, where
\[
A'=
\begin{bmatrix}
1                            & \cdots            & 1                            & 1  \\
\beta_1                      & \cdots            & \beta_k                      & \delta\\
\vdots                       & \ddots            & \vdots                       & \vdots   \\
\beta_1^{k-1}                & \cdots            & \beta_{k}^{k-1}              & \delta^{k-1} \\
\frac{1}{\beta_1-\delta}     & \cdots            & \frac{1}{\beta_k-\delta}     & 0
\end{bmatrix},
A''=
\begin{bmatrix}
1                            & \cdots            & 1                            & 0 \\
\beta_1                      & \cdots            & \beta_k                      & 0 \\
\vdots                       & \ddots            & \vdots                       & 0 \\
\beta_1^{k-1}                & \cdots            & \beta_{k}^{k-1}              & 0 \\
\frac{1}{\beta_1-\delta}     & \cdots            & \frac{1}{\beta_k-\delta}     & \gamma
\end{bmatrix}.
\]
Since
\[
\prod_{i=1}^k(\beta_i-\delta)\det(A') =
\begin{vmatrix}
\beta_1-\delta                                  & \cdots            & \beta_k-\delta                              & 1  \\
\beta_1(\beta_1-\delta)                         & \cdots            & \beta_k(\beta_k-\delta)                     & \delta\\
\vdots                                          & \ddots            & \vdots                                      & \vdots   \\
\beta_1^{k-1}(\beta_1-\delta)                   & \cdots            & \beta_{k}^{k-1}(\beta_k-\delta)             & \delta^{k-1} \\
1                                               & \cdots            & 1                                           & 0 \notag
\end{vmatrix}
\]
\[
=
\begin{vmatrix}
\beta_1                             & \cdots            & \beta_k                         & 1  \\
\beta_1^2                           & \cdots            & \beta_k^2                       & 2\delta \\
\vdots                              & \ddots            & \vdots                          & \vdots   \\
\beta_1^k                           & \cdots            & \beta_{k}^k                     & k\delta^{k-1} \\
1                                   & \cdots            & 1                               & 0 \notag
\end{vmatrix}
\]
\[
\qquad\quad =(-1)^k
\begin{vmatrix}
1                                  & \cdots            & 1                               & 0 \\
\beta_1                            & \cdots            & \beta_k                         & 1  \\
\beta_1^2                          & \cdots            & \beta_k^2                       & 2\delta \\
\vdots                             & \ddots            & \vdots                          & \vdots   \\
\beta_1^k                          & \cdots            & \beta_{k}^k                     & k\delta^{k-1} \notag
\end{vmatrix}
\]
\[
\qquad\quad =(-1)^k
\begin{vmatrix}
1                                  & \cdots            & 1                               & \frac{d}{dx}1\bigg|_{x=\delta} \\
\beta_1                            & \cdots            & \beta_k                         & \frac{d}{dx}x\bigg|_{x=\delta}  \\
\beta_1^2                          & \cdots            & \beta_k^2                       & \frac{d}{dx}x^2\bigg|_{x=\delta} \\
\vdots                             & \ddots            & \vdots                          &  \vdots   \\
\beta_1^k                          & \cdots            & \beta_{k}^k                     & \frac{d}{dx}x^k\bigg|_{x=\delta} \notag
\end{vmatrix}
\]
\[
\qquad\qquad \quad =(-1)^k\frac{d}{dx}
\begin{vmatrix}
1                                 & \cdots            & 1                               & 1 \\
\beta_1                           & \cdots            & \beta_k                         & x  \\
\beta_1^2                         & \cdots            & \beta_k^2                       & x^2 \\
\vdots                            & \ddots            & \vdots                          & \vdots   \\
\beta_1^k                         & \cdots            & \beta_{k}^k                     & x^k \notag
\end{vmatrix}
\bigg|_{x=\delta}
\]
\[
\qquad \qquad \qquad \qquad\qquad \qquad = (-1)^k\frac{d}{dx}\left[\prod_{1\leqslant i<j \leqslant k}(\beta_j-\beta_i)\prod_{i=1}^k(x-\beta_i)\right]\bigg|_{x=\delta},
\]
thus
\[
\begin{aligned}
\det(A') &= \frac{(-1)^k}{\prod_{i=1}^k(\beta_i-\delta)}\prod_{1\leqslant i<j\leqslant k}(\beta_j-\beta_i)\frac{d}{dx}\left[\prod_{i=1}^k(x-\beta_i)\right]\bigg|_{x=\delta}\\
         &= \frac{(-1)^k}{\prod_{i=1}^k(\beta_i-\delta)}\prod_{1\leqslant i<j\leqslant k}(\beta_j-\beta_i)\prod_{i=1}^k(\delta-\beta_i)\sum_{i=1}^k\frac{1}{\delta-\beta_i}\\
         &= \prod_{1\leqslant i<j\leqslant k}(\beta_j-\beta_i)\sum_{i=1}^k\frac{1}{\delta-\beta_i}.
\end{aligned}
\]
It follows that
\[
\begin{aligned}
\det(A) & = \det(A') + \det(A'') \\
        & = \prod_{1\leqslant i<j\leqslant k}(\beta_j-\beta_i)\sum_{i=1}^k\frac{1}{\delta-\beta_i}+\gamma\prod_{1\leqslant i<j \leqslant k}(\beta_j-\beta_i)
\end{aligned}
\]
Hence $\det(A)=0$ is equivalent to
\[
\sum_{i=1}^k\frac{1}{\delta-\beta_i}+\gamma=0.
\]

Designate the set $\{\frac{1}{\delta-\beta_i}|i\in D_{d}\}$ as $S_1$ with cardinality $k+d$. Since $\frac{p-1}{2}\leqslant k,2\leqslant d,$ from Theorem \ref{SH}, we conclude that
\[
\begin{aligned}
|k ^\wedge S_1| &\geqslant \min\{p,k|S_1|-k^2+1\} \\
                &=p,
\end{aligned}
\]
which implies that for each $\gamma \in \F_p$, there exists a subset $\{\beta_1,\ldots,\beta_k\}\subset D_{k}$ such that $\sum_{i=1}^k\frac{1}{\delta-\beta_i}+\gamma=0.$
\end{proof}

\begin{lemma} \label{case_2}
Let $p$ be an odd prime, $k\geqslant \frac{p-1}{2}, d\geqslant 2$ be a positive integer and $D_{d+1}=\{\alpha_1,\ldots,\alpha_{k+d+1}=\delta\}\subset\F_p$.
For any $\delta'\in\F_p,\delta'\notin D_{d+1},\gamma\in\F_p,\gamma\neq \frac{1}{\delta-\delta'}$, there exists a subset $\{\beta_1,\ldots,\beta_k\}\subset D_{d+1}\setminus\{\delta\}$ such that the matrix
\[
B=
\begin{bmatrix}
1                            & \cdots            & 1                            & 1  \\
\beta_1                      & \cdots            & \beta_k                      & \delta\\
\vdots                       & \ddots            & \vdots                       & \vdots   \\
\beta_1^{k-1}                & \cdots            & \beta_{k}^{k-1}              & \delta^{k-1} \\
\frac{1}{\beta_1-\delta'}    & \cdots            & \frac{1}{\beta_k-\delta'}    & \gamma
\end{bmatrix}
\]
is singular.
\end{lemma}
\begin{proof}
Note that $\det(B)=\det(B')+\det(B'')$, where
\[
B'=
\begin{bmatrix}
1                            & \cdots            & 1                            & 1  \\
\beta_1                      & \cdots            & \beta_k                      & \delta\\
\vdots                       & \ddots            & \vdots                       & \vdots   \\
\beta_1^{k-1}                & \cdots            & \beta_{k}^{k-1}              & \delta^{k-1} \\
\frac{1}{\beta_1-\delta'}    & \cdots            & \frac{1}{\beta_k-\delta'}    & \frac{1}{\delta-\delta'}
\end{bmatrix},
B''=
\begin{bmatrix}
1                            & \cdots            & 1                            & 0 \\
\beta_1                      & \cdots            & \beta_k                      & 0 \\
\vdots                       & \ddots            & \vdots                       & 0 \\
\beta_1^{k-1}                & \cdots            & \beta_{k}^{k-1}              & 0 \\
\frac{1}{\beta_1-\delta'}    & \cdots            & \frac{1}{\beta_k-\delta'}    & \gamma-\frac{1}{\delta-\delta'}
\end{bmatrix}.
\]
Since
\[
\qquad \qquad \quad \prod_{i=1}^k(\beta_i-\delta)\det(B') =
\begin{vmatrix}
\beta_1-\delta'                                  & \cdots            & \beta_k-\delta'                                  & \delta-\delta'  \\
\beta_1(\beta_1-\delta')                         & \cdots            & \beta_k(\beta_k-\delta')                         & \delta(\delta-\delta')\\
\vdots                                           & \ddots            & \vdots                                           & \vdots   \\
\beta_1^{k-1}(\beta_1-\delta')                   & \cdots            & \beta_{k}^{k-1}(\beta_k-\delta')                 & \delta^{k-1}(\delta-\delta') \\
1                                                & \cdots            & 1                                                & 1 \notag
\end{vmatrix}
\]
\[
=
\begin{vmatrix}
\beta_1                             & \cdots            & \beta_k                         & \delta   \\
\beta_1^2                           & \cdots            & \beta_k^2                       & \delta^2 \\
\vdots                              & \ddots            & \vdots                          & \vdots   \\
\beta_1^k                           & \cdots            & \beta_{k}^k                     & \delta^{k} \\
1                                   & \cdots            & 1                               & 1 \notag
\end{vmatrix}
\]
\[
\qquad =(-1)^k
\begin{vmatrix}
1                               & \cdots                & 1                         & 1       \\
\beta_1                         & \cdots                & \beta_k                   & \delta  \\
\beta_1^2                       & \cdots                & \beta_k^2                 & \delta^2 \\
\vdots                          & \ddots                & \vdots                    & \vdots   \\
\beta_1^k                       & \cdots                & \beta_{k}^k               & \delta^{k} \notag
\end{vmatrix}
\]
\[
\qquad \qquad \qquad \qquad =(-1)^k\prod_{1\leqslant i<j \leqslant k}(\beta_j-\beta_i)\prod_{i=1}^k(\delta-\beta_i),
\]
we have
\[
\begin{aligned}
\det(B') &= \frac{(-1)^k}{(\delta-\delta')\prod_{i=1}^{k}(\beta_i-\delta')}\prod_{1\leqslant i < j \leqslant j}(\beta_j-\beta_i)\prod_{i=1}^{k}(\delta-\beta_i) \\
         &= \frac{1}{\delta-\delta'}\prod_{1\leqslant i<j \leqslant k}(\beta_j-\beta_i)\prod_{i=1}^{k}\frac{\beta_i-\delta}{\beta_i-\delta'},
\end{aligned}
\]
and
\[
\det(B'')=(\gamma-\frac{1}{\delta-\delta'})\prod_{1\leqslant i<j \leqslant k}(\beta_j-\beta_i).
\]
Hence
\[
\begin{aligned}
\det(B)  &= \frac{1}{\delta-\delta'}\prod_{1\leqslant i<j \leqslant k}(\beta_j-\beta_i)\prod_{i=1}^{k}\frac{\beta_i-\delta}{\beta_i-\delta'}+(\gamma-\frac{1}{\delta-\delta'})\prod_{1\leqslant i<j \leqslant k}(\beta_j-\beta_i) \\
         &= \prod_{1\leqslant i<j \leqslant k}(\beta_j-\beta_i)\left[\frac{1}{\delta-\delta'}\prod_{i=1}^{k}\frac{\beta_i-\delta}{\beta_i-\delta'}+\frac{\gamma(\delta-\delta')-1}{\delta-\delta'}\right]\\
         &= \frac{\prod_{1\leqslant i<j \leqslant k}(\beta_j-\beta_i)}{\delta-\delta'}\left[\prod_{i=1}^{k}\frac{\beta_i-\delta}{\beta_i-\delta'}+\gamma(\delta-\delta')+1\right].
\end{aligned}.
\]

It follows that $\det(B)=0$ is equivalent to
\[
\prod_{i=1}^{k}(1+\frac{\delta'-\delta}{\beta_i-\delta'})=1-\gamma(\delta-\delta').
\]

If $|D_d|=k+2$, we consider the ``dual" version of the equality.
From Corollary \ref{2product}, there exist two distinct elements $x,y\in D_{d}$ such that $(1+\frac{\delta'-\delta}{x-\delta'})(1+\frac{\delta'-\delta}{y-\delta'})=\theta$
for any $\theta\in \F_p^*$, hence there exist $k$ distinct elements in $D_{d}$ such that
\[
\prod_{i=1}^{k}(1+\frac{\delta'-\delta}{\alpha_i-\delta'})=1-\gamma(\delta-\delta'),
\]
for any $\gamma\neq \frac{1}{\delta-\delta'}$.

If $|D_{d}|>k+2$, we select a subset $D'\subset D_{d}$ such that $|D'|=k+2$, then apply the same argument as above.
\end{proof}

\begin{lemma} \label{case_3}
Let $p$ be an odd prime, $k\geqslant \frac{p-1}{2}, d\geqslant 2$ be a positive integer and $D_{d+1}=\{\alpha_1,\ldots,\alpha_{k+d+1}=\delta\}\subset\F_p$.
For any $\gamma\in\F_p,\gamma\neq \delta^k$, there exists a subset $\{\beta_1,\ldots,\beta_k\}\subset D_{d+1}\setminus\{\delta\}$ such that the matrix
\[
B=
\begin{bmatrix}
1                            & \cdots            & 1                            & 1  \\
\beta_1                      & \cdots            & \beta_k                      & \delta\\
\vdots                       & \ddots            & \vdots                       & \vdots   \\
\beta_1^{k-1}                & \cdots            & \beta_{k}^{k-1}              & \delta^{k-1} \\
\beta_1^k				     & \cdots            & \beta_k^k    & \gamma
\end{bmatrix}
\]
is singular.
\end{lemma}

\begin{proof}
Note that $\det(B)=\det(B')+\det(B'')$, where
\[
B'=
\begin{bmatrix}
1                            & \cdots            & 1                            & 1  \\
\beta_1                      & \cdots            & \beta_k                      & \delta\\
\vdots                       & \ddots            & \vdots                       & \vdots   \\
\beta_1^{k-1}                & \cdots            & \beta_{k}^{k-1}              & \delta^{k-1} \\
\frac{1}{\beta_1-\delta'}    & \cdots            & \frac{1}{\beta_k-\delta'}    & \delta^k
\end{bmatrix},
B''=
\begin{bmatrix}
1                            & \cdots            & 1                            & 0 \\
\beta_1                      & \cdots            & \beta_k                      & 0 \\
\vdots                       & \ddots            & \vdots                       & 0 \\
\beta_1^{k-1}                & \cdots            & \beta_{k}^{k-1}              & 0 \\
\frac{1}{\beta_1-\delta'}    & \cdots            & \frac{1}{\beta_k-\delta'}    & \gamma-\delta^k
\end{bmatrix}.
\]
Since
\[
\det(B') = \prod_{1\leqslant i < j \leqslant k}(\beta_j-\beta_i)\prod_{i=1}^{k}(\delta-\beta_i),
\]
\[
\det(B'')= \prod_{1\leqslant i < j \leqslant k}(\beta_j-\beta_i)(\gamma-\delta^k),
\]
we have
\[
\frac{1}{\prod_{1\leq i < j \leq k}(\beta_j-\beta_i)}\det(B) =  \prod_{i=1}^{k}(\delta-\beta_i) + \gamma-\delta^k.
\]
Thus $\det(B)=0$ is equivalent to
\[
\prod_{i=1}^{k}(\delta-\beta_i) = \delta^k - \gamma.
\]
Consider $S=\{\delta-\alpha|\alpha\in D\}$.
If $|D_d|=k+2$, we consider the ``dual" version of the equality.
From Corollary \ref{2product}, there exist two distinct elements $x,y\in S$ such that $xy=\delta^k-\gamma$
for any $\delta^k-\gamma\in \F_p^*$, hence there exist $k$ distinct elements in $S$ such that
\[
\prod_{i=1}^{k}(\delta-\beta_i) = \delta^k - \gamma,
\]
for any $\gamma\neq \delta^k$.

If $|D_{d}|>k+2$, we select a subset $S'\subset S$ such that $S=k+2$, then apply the same argument as above. 
\end{proof}

Now we prove Theorem \ref{mainthm}.
\begin{proof} (of Theorem \ref{mainthm}) Proceed by induction on the depth of the full deep hole tree.

\noindent
\textbf{Basis case} This follows from Lemma \ref{basis}.

\noindent
\textbf{Inductive step}
We need to show that if the set of nodes of the full deep hole tree coincide with the nodes of the expected deep hole tree in the same depth $d \geqslant 2$, then there are no additional nodes in depth $d+1$ except the expected ones. Denote the corresponding evaluation set by $D_{d}=\{\alpha_1,\ldots,\alpha_{k+d}\}$ in depth $d$ and $D_{d+1}=\{\alpha_1,\ldots,\alpha_{k+d},\alpha_{k+d+1}=\delta \}$ in depth $d+1$. In order to show there are no new nodes in depth $d+1$, There are three cases to consider.

\noindent
\textbf{Case 1:} We need to show the branch, which is corresponding to the function $f=\frac{1}{x-\delta}$, will not continue in the depth $d+1$.
It suffices to show that there exists a subset $\{\beta_1,\ldots,\beta_k\}\subset \{\alpha_1,\ldots,\alpha_{k+d}\}$ such that for any $\gamma\in \F_p$ and matrix
\[
A=
\begin{bmatrix}
1                           & \cdots            & 1                            & 1  \\
\beta_1                     & \cdots            & \beta_k                      & \delta\\
\vdots                      & \ddots            & \vdots                       & \vdots   \\
\beta_1^{k-1}               & \cdots            &  \beta_{k}^{k-1}             & \delta^{k-1} \\
\frac{1}{\beta_1-\delta}    & \cdots            & \frac{1}{\beta_k-\delta}     & \gamma
\end{bmatrix}
\]
we have $\det(A)=0$. This follows from Lemma \ref{case_1}.

\noindent
\textbf{Case 2:} We need to show that the branch, which is corresponding to the function $f=\frac{1}{x-\delta'}$, where $\delta' \notin D_{k+1}$, has only one child in depth $d+1$.
It suffices to show that there exists a subset $\{\beta_1,\ldots,\beta_k\}\subset D_{d}$ such that for any $\delta'\notin D_{d+1},\gamma\in\F_p,\gamma\neq \frac{1}{\delta-\delta'}$ and matrix
\[
B=
\begin{bmatrix}
1                            & \cdots            & 1                            & 1  \\
\beta_1                      & \cdots            & \beta_k                      & \delta\\
\vdots                       & \ddots            & \vdots                       & \vdots   \\
\beta_1^{k-1}                & \cdots            & \beta_{k}^{k-1}              & \delta^{k-1} \\
\frac{1}{\beta_1-\delta'}    & \cdots            & \frac{1}{\beta_k-\delta'}    & \gamma
\end{bmatrix}
\]
we have $\det(B)=0$. This follows from Lemma \ref{case_2}.

\noindent
\textbf{Case 3:} We need to show that the branch, which is corresponding to the function $f=x^k$ has only one child in each depth.
It suffices to show that there exists a subset $\{\beta_1,\ldots,\beta_k\}\subset D_{d}$ such that for any $\gamma\neq \delta^k$ and matrix
\[
B=
\begin{bmatrix}
1                            & \cdots            & 1                            & 1  \\
\beta_1                      & \cdots            & \beta_k                      & \delta\\
\vdots                       & \ddots            & \vdots                       & \vdots   \\
\beta_1^{k-1}                & \cdots            & \beta_{k}^{k-1}              & \delta^{k-1} \\
\beta_1^k			         & \cdots            & \beta_k^k				    & \gamma
\end{bmatrix}
\]
we have $\det(B)=0$. This follows from Lemma \ref{case_3}.

From the principle of induction, the theorem is proved.
\end{proof}

\section{Proof of Theorem \ref{geometrythm}}
\begin{proof} There are 3 cases to prove.

\noindent
\textbf{Case 1.} Let $RS_q(F_q, k)$ be an extended GRS code over the finite
field $\F_q$ whose characteristic $p$ is odd. Let one of its generator matrix be
$$
G=
\begin{bmatrix}
1           & 1          & \cdots & 1 \\
\alpha_1    & \alpha_2   & \cdots & \alpha_q\\
\alpha_1^2  & \alpha_2^2 & \cdots & \alpha_q^2\\
\vdots      & \vdots     & \ddots & \vdots \\
\alpha_1^{k-1} & \alpha_2^{k-1} & \cdots & \alpha_q^{k-1}
\end{bmatrix},
$$
where $\alpha_1,\ldots,\alpha_q$ are distinct element of  $\F_q$.

Suppose a word $u\in \F_q^q$ is a deep hole of $RS_q(F_q, k)$. From proposition \ref{prop}, this is equivalent to the fact that
$$
G'=\left[
\begin{array}{c}
 G\\ \hline
 u
\end{array}\right]
$$
generates another linear MDS code, where 
\[
u=(u_1,u_2,\ldots,u_q).
\]
Thus the set
$$S=\{c_1,\ldots,c_q\}\cup\{(0,\ldots,0,1)\},$$
where $c_i$ is the $i$-th column of $G'$ for $1\leqslant i \leqslant q$, has size $q+1$ and has the property that every subset of $S$ of size $k+1$ is a basis.

Since $k+1\leqslant p$ or $3\leqslant q-p+1\leqslant k+1\leqslant q-2$, by Theorem \ref{Simeon}, we deduce that $S$ is equivalent to the set
$$\{(1,\alpha,\alpha^2,\ldots,\alpha^k)\mid \alpha\in \F_q\}\cup\{(0,\ldots,0,1)\}.$$

Thus we conclude that
$$u(x)=ax^k+f_{\leqslant k-1}(x), a\neq 0;$$
where $f_{\leqslant k-1}(x)$ denotes a polynomial with degree not larger
than $k-1$.

\noindent
\textbf{Case 2.} Firstly, we get an estimation of $N_{min}(k,q)$. Combining Theorem \ref{B.Segre} and Lemma \ref{RL}, we conclude that
$$
\begin{aligned}
N_{min}(k,q) &\leqslant N_{min}(3,q)+k-3\\
             &\leqslant \lceil{q-\dfrac{\sqrt{q}-7}{4}}\rceil +k-3 \\
             &\leqslant q-1.
\end{aligned}
$$

Now let $G$ be a generator matrix of $RS_q(F_q^*,k)$ of the following form
$$
G=
\begin{bmatrix}
1           & 1          & \cdots & 1 \\
\alpha_1    & \alpha_2   & \cdots & \alpha_{q-1}\\
\alpha_1^2  & \alpha_2^2 & \cdots & \alpha_{q-1}^2\\
\vdots      & \vdots     & \ddots & \vdots \\
\alpha_1^{k-1} & \alpha_2^{k-1} & \cdots & \alpha_{q-1}^{k-1}
\end{bmatrix},
$$
where $\alpha_1,\ldots,\alpha_{q-1}$ are distinct element of  $\F_q^*$.
From proposition \ref{prop}, a word $u\in \F_q^{q-1}$ is a deep hole of $RS_q(F_q^*, k)$ if and only if
$$
G'=\left[
\begin{array}{c}
 G\\ \hline
 u
\end{array}\right]
$$
generates another linear MDS code $\C_2$, where
\[
u=(u_1,u_2,\ldots,u_{q-1}).
\]
Since $\C_2$ is of length $q-1$, thus the matrix $G'$ is equivalent to a Vandermonde matrix of rank $k$.
Notice that $G$ is the given Vandermonde matrix of rank $k-1$. Thus there are two possibilities of $u$, i.e., its Lagrange interpolation
polynomial satisfies the following conditions:
$$u(x)=ax^k+f_{\leqslant k-1}(x), a\neq 0;$$
or
$$u(x)=bx^{q-2}+f_{\leqslant k-1}(x), b\neq 0;$$
where $f_{\leqslant k-1}(x)$ denotes a polynomial with degree not larger
than $k-1$.

\noindent
\textbf{Case 3.} This is similar with the proof of case 2 and we will make use of Theorem \ref{Voloch}.
\end{proof}

\section{Concluding Remarks}

In this paper,  we classify deep holes completely for Generalized Reed-Solomon codes $RS_p(D,k)$, where
$p$ is a prime, $|D| > k \geqslant \frac{p-1}{2}$. We suspect that a similar result
hold over finite fields of composite order, and
leave it as an open problem.

\bibliographystyle{plain}
\bibliography{comp_deep_hole}
\appendix
\section{Proof of Proposition \ref{prop}}
\begin{proof}
$\Rightarrow$ Suppose $u$ is a deep hole of $C=[n,k]_q$, we need to show that $G'$ is a generator matrix for another MDS code.
Equivalently, we need to show that any $k+1$ columns of $G'$ are linearly independent.

Assume there exist $k+1$ columns of $G'$ which are linearly dependent.  Without loss of generality, we assume that the first $k+1$ columns of $G'$ are linear dependent. Consider the submatrix consisting of the intersection of the first $k+1$ rows and the first $k+1$ columns of $G'$. Hence there exist $a_1,\ldots,a_k\in\F_q$, not all zero, such that
$$(u_1,\ldots,u_{k+1})=a_1r_{1,k+1}+\cdots+a_kr_{k,k+1},$$
where $r_{i,k+1}$ is the vector consisting of the first $k+1$ elements of the $i$-th row of $G$ for $1\leqslant i \leqslant k$. Let $v=a_1r_1+\cdots+a_kr_k\in C$, where $r_i$ is the $i$-th row of $G$ for $1\leqslant i \leqslant k$. We have
$$d(u,v)\leqslant n-(k+1)<\rho,$$
which is a contradiction with the assumption that $u$ is a deep hole of $C$.

$\Leftarrow$ Now suppose $G'$ is a generator matrix for an MDS code, i.e., any $k+1$ columns of $G'$ are linearly independent.
We need to show that $d(u,C)=n-k$.

Assume that $d(u,C)<n-k.$
Equivalently,
there exist $a_1,\ldots,a_k\in\F_q$ such that $u$ and $v=a_1r_1+\cdots+a_kr_k$ have more than $k$ common coordinates, where $r_i$ is the $i$-th row of $G$ for $1\leqslant i \leqslant k$. Without loss of generality, we assume that the first $k+1$ coordinates of $u$ and $v$ are the same. Consider the submatrix consisting of the first $k+1$ columns. Since the rank of the matrix is less than $k+1$,
thus the first $k+1$ columns of $G'$ are linearly dependent, which contradicts the assumption.
\end{proof}





\end{document}